\documentclass[%
reprint,
nofootinbib,
 amsmath,amssymb,
 aps,
 pra,
showkeys,
]{revtex4-2}

\usepackage{mathrsfs}
\usepackage{amsfonts}
\usepackage{amsmath} 
\usepackage{amssymb} 

\usepackage{amsthm}
\usepackage{xspace}
\usepackage[mathcal]{euscript}

\usepackage[export]{adjustbox}
\usepackage{dcolumn}
\usepackage{bm}

\usepackage{orcidlink}

\usepackage{graphicx}
\usepackage{tikz}
\usetikzlibrary{calc,arrows}
\usetikzlibrary{decorations.text}
\usepackage{twemojis}

\usepackage{url}
\usepackage{hyperref}
\usepackage{color}
\definecolor{refcolorCite}{RGB}{130,0,50}
\definecolor{refcolorLink}{RGB}{0,0,150}
\definecolor{refcolorURL}{RGB}{0,0,190}
\hypersetup{
    colorlinks,
    citecolor=refcolorCite,
    filecolor=refcolorCite,
    linkcolor=refcolorLink,
    urlcolor=refcolorURL
}

\usepackage{booktabs}

\usepackage{bookmark}
\bookmarksetup{
  numbered, 
  open,
}

\usepackage{natbib}
\setcitestyle{round,authoryear}

\newtheorem{theorem}{Theorem}
\newtheorem{question}{Question}

\theoremstyle{definition}

\newtheorem{lemma}{Lemma}

\newtheorem{remark}{Remark}

\newcommand{\ThesisName}{Thesis:}
\newtheorem{tagThesisx}{\ThesisName}
\newenvironment{tagThesis}[1]
 {\renewcommand\thetagThesisx{#1}\tagThesisx}
 {\endtagThesisx}

\newcommand{\DefinitionName}{Definition:}
\newtheorem{tagDefinitionx}{\DefinitionName}
\newenvironment{tagDefinition}[1]
 {\renewcommand\thetagDefinitionx{#1}\tagDefinitionx}
 {\endtagDefinitionx}

\theoremstyle{remark}

\numberwithin{proofStep}{theorem} 

\newcommand{\orcid}[1]{\href{https://orcid.org/#1}{\textcolor[HTML]{A6CE39}{\aiOrcid}}}

\def\({\big(}
\def\){\big)}

\newcommand{\tn}{\textnormal}

\newcommand{\hilbert}{\mathcal{H}}

\newcommand{\ms}[1]{\mathscr{#1}}

\newcommand{\wh}[1]{\widehat{#1}}

\newcommand{\oper}[1]{\wh{\mathbf{#1}}}

\newcommand{\R}{\mathbb{R}}

\newcommand{\abs}[1]{\left\lvert#1\right\rvert}

\newcommand{\de}{\operatorname{d}}

\newcommand{\eg}{\textit{e.g.}\ }

\newcommand{\etc}{\textit{etc}}

\newcommand{\bra}[1]{\left\langle#1\right|}
\newcommand{\ket}[1]{\left|#1\right\rangle}
\newcommand{\braket}[2]{\langle#1|#2\rangle}

\def\sref #1{\S\ref{#1}}

\newcommand{\icoPrince}
{\raisebox{-0.75em}{\twemoji[scale=0.75]{chart increasing}}}
\newcommand{\icoPauper}
{\raisebox{-0.75em}{\twemoji[scale=0.75]{chart decreasing}}}

\begin{document}


\title{The prince and the pauper\\A quantum paradox of Hilbert-space fundamentalism}

\author{Ovidiu Cristinel Stoica\ \orcidlink{0000-0002-2765-1562}}
\affiliation{
 Dept. of Theoretical Physics, NIPNE---HH, Bucharest, Romania. \\
	Email: \href{mailto:cristi.stoica@theory.nipne.ro}{cristi.stoica@theory.nipne.ro},  \href{mailto:holotronix@gmail.com}{holotronix@gmail.com}
	}%

\date{\today}

\begin{abstract}
The quantum world is described by a unit vector in the Hilbert space and the Hamiltonian. Do these abstract basis-independent objects give a complete description of the physical world, or should we include observables like positions and momenta and the decomposition into subsystems? According to ``Hilbert-space fundamentalism'' they give a complete description, and all other features of the physical world emerge from them \citep{Carroll2021RealityAsAVectorInHilbertSpace}. Here I will give a concrete refutation of this thesis based on the symmetries of the theory of quantum measurements. These results show that even if a tensor product structure is assumed along with the unit vector and the Hamiltonian, concrete physically distinct worlds can be described by the same structures.
\end{abstract}


\maketitle

\section{Introduction}
\label{s:intro}

A quantum system, which may be the entire world, is represented by a unit vector $\ket{\psi(t)}$ called \emph{state vector}. $\ket{\psi(t)}$ belongs to a \emph{state space} $\hilbert$, a complex vector space endowed with a scalar product $\braket{\psi}{\xi}=\braket{\xi}{\psi}^\ast$, having some continuity properties that make it a \emph{Hilbert space}.

A system in the state $\ket{\psi(0)}$ changes, after a time interval $t$, according to the \emph{evolution equation}:
\begin{equation}
\label{eq:unitary_evolution}
\ket{\psi(t)}=\oper{U}_t\ket{\psi(0)},
\end{equation}
where $\oper{U}_t=e^{-i/\hbar\oper{H}t}$, $\hbar$ is the \emph{reduced Planck constant}, and $\oper{H}$ is the \emph{Hamiltonian operator}. $\oper{H}$ is time independent for closed systems, even for the entire universe. 
The \emph{evolution operators} $\oper{U}_t$ preserve the complex vector space structure and the scalar product, so they are \emph{unitary}.

\newcommand\lrefBQS{\hyperref[def:BQS]{basic quantum structure}\xspace}
\newcommand\lrefBQSs{\hyperref[def:BQS]{basic quantum structures}\xspace}
\begin{tagDefinition}{Basic quantum structures}
\label{def:BQS}
We call the triple $(\hilbert,\oper{H},\ket{\psi(t)})$ \emph{basic quantum structure}.
\end{tagDefinition}

This quantum formalism raises the following problem:
\begin{question}
\label{question:BSF-completeness}
Does the \lrefBQS $(\hilbert,\oper{H},\ket{\psi(t)})$ give a complete description of reality?
\end{question}

But a unit vector is like any other one, structureless. The structure of the world is manifest in the relation between the otherwise identical unit vectors and the operators representing observables, which encode the physical meaning.
In Quantum Mechanics  we represent different subsystems on distinct Hilbert spaces $\hilbert_1,\hilbert_2,\ldots$, and their physical properties by operators, so that we connect the theory with reality.
And we reflect this in our mathematical notations and in the informal language accompanying them.
The triple $(\hilbert,\oper{H},\ket{\psi(t)})$ is supplemented with the tensor product structure $\hilbert=\hilbert_1\otimes\hilbert_2\otimes\ldots$ and various operators on the spaces $\hilbert_1,\hilbert_2,\ldots$ to represent the observable properties.
The \lrefBQS is treated in practice as insufficient to describe the world without these additional structures. 

But some researchers think that these should not be hard-coded in the formalism, and endorse the following:

\newcommand\lrefHSF{\hyperref[thesis:HSF]{Hilbert-space fundamentalism}\xspace}
\begin{tagThesis}{Hilbert-space fundamentalism}
\label{thesis:HSF}
The tensor product structure and the physical meaning of the observables, and everything in the world, emerge uniquely from the \lrefBQS alone \citep{CarrollSingh2019MadDogEverettianism,Carroll2021RealityAsAVectorInHilbertSpace}.
This thesis, coined ``\lrefHSF'' in \citep{Carroll2021RealityAsAVectorInHilbertSpace}, is assumed in various research programs, \eg in some approaches based on decoherence where along with the pointer state the physical meaning is supposed to emerge
\citep{Zurek1991DecoherenceAndTheTransitionFromQuantumToClassical,Zurek2003DecoherenceEinselectionAndTheQuantumOriginsOfTheClassical,Schlosshauer2007DecoherenceAndTheQuantumToClassicalTransition}, claims that the tensor product structure is uniquely determined by the Hamiltonian's spectrum under conditions like the locality of the interactions \citep{Piazza2010GlimmersOfAPreGeometricPerspective,CotlerEtAl2019LocalityFromSpectrum,CarrollSingh2021QuantumMereology},  proposals that space itself emerges from these abstract structures alone \citep{CarrollSingh2019MadDogEverettianism,Carroll2021RealityAsAVectorInHilbertSpace,CarrollSingh2021QuantumMereology}, and other Quantum Gravity programs in which spacetime is supposed to emerge from the quantum structure \citep{Stoica2022SpaceThePreferredBasisCannotUniquelyEmergeFromTheQuantumStructure}.
\end{tagThesis}

\begin{figure}[!ht]
\centering
\begin{tikzpicture}[scale=1]
\foreach \i in {0,...,11}
{
    \draw[rounded corners=2.5mm,color=violet!50,fill=violet!25,rotate around={\i * 30-60:(2,2)}] (1.75,1.75) rectangle ++(2.0,0.5);
		\draw [very thick,color=red!50!black!50,fill=red!25,rotate around={\i * 30-60:(2,2)}] (3.5,2) circle (1.0mm);
}
\draw [very thick,color=orange!50!black,fill=orange] (2,2) circle (1.0mm);
\node at (2,1.4) {$(\hilbert,\oper{H},\ket{\psi(t)})$};
\draw[->,thick] (2,1.55) -- ++(0,0.25);
\path[postaction={decoration={
                    text along path,
                    text={Underdetermined emerging solutions},
                    text align=center,
                    reverse path},
      decorate}] (3.9,2) arc (0:180:1.9cm);

\begin{scope}[shift={(3.5,0)}]
\draw[rounded corners=2.5mm,color=violet,fill=violet!50] (1.75,1.75) rectangle ++(2.0,0.5);
\draw [very thick,color=orange!50!black,fill=orange] (2,2) circle (1.0mm);
\draw [very thick,color=red!50!black,fill=red] (3.5,2) circle (1.0mm);
\node at (2,1.4) {$(\hilbert,\oper{H},\ket{\psi(t)})$};
\node[align=center] at (3.5,2.9){TPS and \\ Observables};
\draw[->,thick] (2,1.55) -- ++(0,0.25);
\draw[->,thick] (3.5,2.45) -- ++(0,-0.25);
\end{scope}
\end{tikzpicture}
\caption{\textbf{Left:} The \lrefBQS specifies incompletely the physical reality.
\textbf{Right:} But it can be completed by including the tensor product structure and the observables.}
\label{fig:complete}
\end{figure}
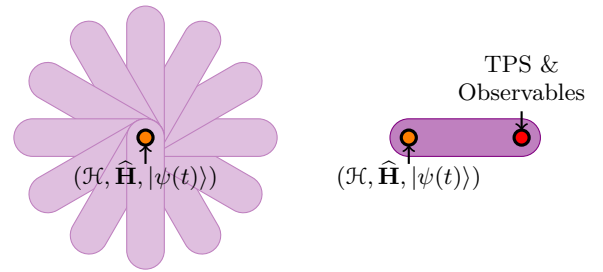

\lrefHSF was refuted in \citep{Stoica2022SpaceThePreferredBasisCannotUniquelyEmergeFromTheQuantumStructure}. In particular, the $3$d space, the tensor product structure, and a preferred basis can't emerge uniquely from the \lrefBQS. Also numerous counterexamples can be found in \citep{Stoica2024EmpiricalAdequacyOfTimeOperatorCC2HamiltonianGeneratingTranslations}.
A proof that the tensor product structure doesn't emerge uniquely from the Hamiltonian was given in \citep{Stoica2024DoesTheHamiltonianDetermineTheTPSAndThe3dSpace}.
The \lrefBQS gives an incomplete description of reality, but it can be completed by including the tensor product structure and the observables (Figure \ref{fig:complete}).

In Section \sref{s:prince-pauper} I show that the same state vector and the same Hamiltonian describe physically distinct realities (not to be confused with the many-worlds in Everett's interpretation!).
This gives a very simple, intuitive and constructive proof that the answer to Question \ref{question:BSF-completeness} is negative (even if we include the tensor product structure). 
In Section \sref{s:classical-level} we will see that this ambiguity extends to the classical level of reality, even for observations that are not explicitly quantum.
Section \sref{s:discussion} concludes with a brief discussion of the implications.
The more technical parts are exiled in Appendix \sref{s:lemmas-proofs}.

\section{The prince and the pauper}
\label{s:prince-pauper}

Edward is a young dreamer with bold ideas, who wants to make positive contributions to the world. He wants to invest, thinking that money would help him achieve his purpose to improve the world. Since he has a risk-embracing attitude, he decides to let quantum measurements make financial decisions for him.
Or maybe he is just practical, not wanting to waste too much time making decisions based on incomplete information.

So whenever he thinks of choosing between two possible investments, or between buying or selling stocks, he lets quantum chance decide for him, and he faithfully bids accordingly. He can do this by making quantum measurement on qubits, or by using Vaidman's \emph{Quantum World Splitter} on the Internet \citep{Vaidman-QuantumWorldSplitter}.

Suppose that if the qubit turns out to be in the state $\ket{+}$ Edward becomes very wealthy. Let's represent the world's state in which Edward is rich like a prince by
\begin{equation}
\label{eq:prince}
\ket{\psi_{+}}=\ket{\icoPrince}.
\end{equation}

But if the opposite result $\ket{-}$ is obtained, he becomes poor like a pauper, and the world's state becomes:
\begin{equation}
\label{eq:pauper}
\ket{\psi_{-}}=\ket{\icoPauper}.
\end{equation}

This scenario is inspired by Mark Twain \citep{MarkTwain1882ThePrinceAndThePauper}.
Based on this scenario, I prove the following result:
\begin{theorem}
\label{thm:BQS-incompleteness}
The same \lrefBQS can represent an unlimited number of physically distinct realities (in which Edward's wealth ranges from very poor to extremely rich).
This is true even if we fix the tensor product structure.
\end{theorem}
\begin{proof}
Two \lrefBQSs $(\hilbert,\oper{H},\ket{\psi(t)})$ and $(\hilbert',\oper{H}',\ket{\phi(t)}')$ are \emph{isomorphic} if there is a unitary operator $\oper{S}:\hilbert\to\hilbert'$ so that
\begin{equation}
\label{eq:BQS-isomorphism}
\begin{cases}
\oper{S}\ket{\psi(t)}=\ket{\phi(t)}'\\
\oper{S}\oper{H}\oper{S}^{-1}=\oper{H}'.\\
\end{cases}
\end{equation}

Then, for any basis $\{\ket{\alpha_1},\ket{\alpha_2},\ldots\}$ of $\hilbert$, for all $j$ and $k$, if $\ket{\beta_j}:=\oper{S}\ket{\alpha_j}$,
\begin{equation}
\label{eq:BQS-isomorphism-basis}
\begin{aligned}
\braket{\alpha_j}{\psi(t)}&=\braket{\beta_j}{\phi(t)}' \\
\bra{\alpha_j}\oper{H}\ket{\alpha_k}&=\bra{\beta_j}\oper{H}'\ket{\beta_k}.
\end{aligned}
\end{equation}

But $\ket{\psi(t)}$ and $\ket{\phi(t)}'$ represent identical physical worlds only if there are such bases made of eigenvectors of operators representing the same physical properties. For example both can be position or momentum eigenbases.
If \lrefHSF is true, such operators should emerge uniquely, and any two isomorphic \lrefBQSs should represent the same reality.

So if we will show that the \lrefBQS $(\hilbert,\oper{H},\ket{\psi_{+}(t)})$ is isomorphic to $(\hilbert,\oper{H},\ket{\psi_{-}(t)})$, where $\ket{\psi_{+}(t)}$ and $\ket{\psi_{-}(t)}$ are two physically distinct worlds, this will refute \lrefHSF.

Let the universe before Edward made the qubit measurement be in the state
\begin{equation}
\label{eq:ready}
\ket{Q}\ket{\tn{ready}},
\end{equation}
where, for simplicity, $\ket{\tn{ready}}$ represents not only the measuring device in its ``ready'' state, but also the entire universe minus the measured qubit.

If $\ket{Q}=\ket{+}$, the evolution law gives
\begin{equation}
\label{eq:ready2up}
\ket{+}\ket{\tn{ready}}\mapsto\ket{\psi_{+}}=\ket{\icoPrince},
\end{equation}
while if $\ket{Q}=\ket{-}$, it gives
\begin{equation}
\label{eq:ready2down}
\ket{-}\ket{\tn{ready}}\mapsto\ket{\psi_{-}}=\ket{\icoPauper}.
\end{equation}

From Lemma \ref{lemma:measurement-transform} (see Appendix \sref{s:lemmas-proofs}), there is a unitary transformation $\oper{S}$ of the total Hilbert space $\hilbert$ that preserves the Hamiltonian $\oper{H}$ -- as in equation \eqref{eq:BQS-isomorphism} -- so that
\begin{equation}
\label{eq:qubit-change-condition-thm}
\oper{S}\ket{\icoPrince}=\ket{\icoPauper}.
\end{equation}

Then, $\oper{S}$ is a \lrefBQS automorphism
\begin{equation}
\label{eq:iso-ambiguity}
\oper{S}:(\hilbert,\oper{H},\ket{\psi_{+}(t)})\to(\hilbert,\oper{H},\ket{\psi_{-}(t)}).
\end{equation}

Therefore, the same \lrefBQS $(\hilbert,\oper{H},\ket{\psi_{+}(t)})$ represents both a world in which Edward is rich, and a world in which he is poor.
According to the transformation from Lemma \ref{lemma:measurement-transform}, both worlds have the same tensor product structure.
All that the transformation $\oper{S}$ did was to invert the outcomes of the measurement.

Moreover, Edward can make unlimitedly many financial decisions by using qubit measurements. Therefore, if these are independent measurements of different qubits, the number of physically distinct alternative worlds represented by the same \lrefBQS has an unlimited exponential growth.

All these worlds may be physically very different, from containing a bankrupt Edward living on the street, to versions of Edward that built various financial empires and failed various businesses, depending on the results of the quantum measurements.
\end{proof}

Therefore, the \lrefBQS doesn't give a complete description of reality, and \lrefHSF is refuted.

\begin{remark}
\label{rem:underlying}
The ambiguity shown in Theorem \ref{thm:BQS-incompleteness} is due to the fact, shown in \citep{Stoica2022SpaceThePreferredBasisCannotUniquelyEmergeFromTheQuantumStructure}, that the additional observables needed for a full description of reality don't emerge uniquely from the \lrefBQS.
\qed
\end{remark}

\section{Ambiguity at the classical level}
\label{s:classical-level}

While the world is described by quantum theory, it appears to us classical, at least as long as we don't make quantum measurements.
But quantum measurements are ubiquitous, for example sight works like a quantum measurement. When we observe visually the positions and shapes of objects, sight works like a position measurement. When we observe color, it works like a momentum measurement, since the wavelength of light is proportional to the momentum.

Both the emergence of classicality at the macro level and the quantum measurements work in the same way. A quantum measurement leads to a superposition of states in which the pointer has observably distinct states. Whatever resolves this superposition, whether it is the wavefunction collapse or decoherence or another mechanism, it also ensures that at the macro level the superposition is gone, and the world appears to us classical.

The observables that represent positions and momenta have continuous spectra, the full set of real numbers $\R$. This is in contrast with the qubit observables, which have only two eigenvalues.
Lemma \ref{lemma:measurement-transform-swap} shows that the resulting states corresponding to two distinct eigenvalues are related by a unitary transformation that preserves the evolution law.

Therefore, the possible states of the macro level of reality that result from the same initial state vector can be described by the same \lrefBQS.
This provides, again, an unlimited number of concrete counterexamples to the \lrefHSF thesis.
These counterexamples don't even require explicitly quantum measurements, only naked eye observations of the world.

\section{Discussion}
\label{s:discussion}

If \lrefHSF were true, space, fields on space, the decomposition into subsystems, a preferred basis, and every other physical feature of the world would emerge uniquely from the state vector and the Hamiltonian.
Any isomorphic \lrefHSF would represent identical physical realities. 

In \citep{Stoica2022SpaceThePreferredBasisCannotUniquelyEmergeFromTheQuantumStructure} it was shown that whenever any of these structures emerges from the state vector and the Hamiltonian, it is not unique.
The only structures that can emerge uniquely can't exhibit physical differences, not even in relation with the state vector.
But the tensor product structure and the $3$d space and other structures of interest exhibit such differences. For example the wavefunction changes with respect to space and to the tensor product structure.
In addition, a specific proof refuting the claim that the tensor product structure emerges uniquely from the Hamiltonian's spectrum \citep{CotlerEtAl2019LocalityFromSpectrum,CarrollSingh2021QuantumMereology,CarrollSingh2019MadDogEverettianism,Carroll2021RealityAsAVectorInHilbertSpace} was given in \citep{Stoica2024DoesTheHamiltonianDetermineTheTPSAndThe3dSpace}, and numerous other counterexamples to \lrefHSF were given in \citep{Stoica2024EmpiricalAdequacyOfTimeOperatorCC2HamiltonianGeneratingTranslations}.

These results were shown to affect all theories that assume an affirmative answer to Question \ref{question:BSF-completeness}, whether they rely on state vector reduction  or branching (\eg the version of Everett's Interpretation coined by Carroll and Singh ``Mad-dog Everettianism''), proposals based on decoherence, and proposals that spacetime emerges from a purely quantum theory of gravity.
\emph{This doesn't mean that such approaches are useless, just that they can't give a complete description of reality.}

The proof given in \citep{Stoica2022SpaceThePreferredBasisCannotUniquelyEmergeFromTheQuantumStructure} is fully general, but it was largely ignored, maybe because it's quite abstract, using tensors on the Hilbert space and invariants, and because it's an existence proof without many constructive counterexamples.
The most intuitive and constructive counterexample given in \citep{Stoica2022SpaceThePreferredBasisCannotUniquelyEmergeFromTheQuantumStructure} uses transformations of the form $\oper{S}=\oper{U}_t$, which preserve the Hamiltonian and its relation with the state vector. This implies that the present time state vector also describes the past and future states of the world.
In  \citep{Stoica2022SpaceThePreferredBasisCannotUniquelyEmergeFromTheQuantumStructure} it was shown that there are infinitely many continuous families of unitary transformations that commute with the Hamiltonian and preserve the state vector, but they were just proven to exist, without showing how we can construct them.
By contrast, the constructions presented in this article give many more intuitive and constructive examples of alternative realities represented by the same state vector and Hamiltonian, and even the same tensor product structure.

\appendix

\section{Proofs of the Lemmas}
\label{s:lemmas-proofs}

Let us recall the \emph{standard model of quantum measurements} (see \eg \citeauthor{Mittelstaedt2004InterpretationOfQMAndMeasurementProcess} \citeyear{Mittelstaedt2004InterpretationOfQMAndMeasurementProcess}, \S 2.2(b), and \citeauthor{BuschGrabowskiLahti1995OperationalQuantumPhysics} \citeyear{BuschGrabowskiLahti1995OperationalQuantumPhysics}, \S II.3.4).
Realistic examples of such measurements are described in (\citeauthor{BuschGrabowskiLahti1995OperationalQuantumPhysics} \citeyear{BuschGrabowskiLahti1995OperationalQuantumPhysics}, \S VII), including spin measurements using the Stern-Gerlach device, photon polarization measurements, various photon counters and beam splitter experiments \etc.

In the case of spin measurements using the Stern-Gerlach apparatus, the possible outcomes of the measurement are distinguished by the region of a photographic plate hit by the observed particle.
The pointer observable corresponds to position.
Therefore, we choose a pointer operator $\oper{Z}$ with the spectrum equal to $\R$.
By working in the interaction picture, we can take the free Hamiltonians of the two systems to be zero, without loss of generality. This allows us to focus only on the interaction Hamiltonian.
The Hamiltonian is
\begin{equation}
\label{eq:measurement_hamiltonian}
\oper{H}=\oper{H}_{\tn{int}} = -g \oper{A}\otimes\oper{p}_{\mathbf{Z}},
\end{equation}
where $\oper{A}$ is the observable, $\oper{p}_{\mathbf{Z}}$ is the canonical conjugate of the pointer operator $\oper{Z}$.
The coupling $g$ is constant in the interval $\left[0,T\right]$ and negligible outside this interval.

Let $\{\ket{\lambda,a}\}_{a\in\ms{A}}$ be a set of orthonormal eigenvectors of $\oper{A}$ corresponding to the eigenvalue $\lambda$. To account for the possible degeneracy of the eigenvalues, they are indexed by a label $a\in\ms{A}$. Since all eigenspaces of $\oper{A}$ have the same dimension, we can choose the same set $\ms{A}$ for all $\lambda$. 

Since $\oper{p}_{\mathbf{Z}}$ is the canonical conjugate of $\oper{Z}$, \begin{equation}
\label{eq:pointer_ccr}
[\oper{Z},\oper{p}_{\mathbf{Z}}]=i\hbar\oper{I}.
\end{equation}

Let $\ket{\zeta}$ be the pointer eigenvector corresponding to the eigenvalue $\zeta\in\R$ of $\oper{Z}$.
The operator $\oper{p}_{\mathbf{Z}}$ generates, for any $\tau\in\R$, the translation
\begin{equation}
\label{eq:pointer_translation}
e^{-i \oper{p}_{\mathbf{Z}} \tau}\ket{\zeta}=\ket{\zeta + \tau}.
\end{equation}

Then, for any eigenvector $\ket{\lambda,a}$ of $\oper{A}$ and any time interval $t\in[0,T]$, we obtain (\citeauthor{Mittelstaedt2004InterpretationOfQMAndMeasurementProcess} \citeyear{Mittelstaedt2004InterpretationOfQMAndMeasurementProcess}, \S 2.2(b)),
\begin{equation}
\label{eq:unitary_evolution_interaction}
\begin{aligned}
\oper{U}_t\ket{\lambda,a}\ket{\zeta} 
&= e^{-\frac{i}{\hbar}\oper{H}t}\ket{\lambda,a}\ket{\zeta} \\
&= e^{\frac{i}{\hbar}g \oper{A}\otimes\oper{p}_{\mathbf{Z}}t}\ket{\lambda,a}\ket{\zeta} \\
&=\ket{\lambda,a}e^{i g t \lambda \oper{p}_{\mathbf{Z}}}\ket{\zeta} \\
&=\ket{\lambda,a}\ket{\zeta - g t \lambda}.
\end{aligned}
\end{equation}

If the ready pointer state is calibrated to be $\ket{0}$ and the resulting pointer state after the time interval $T$ is $\ket{- g T \lambda}$, the corresponding eigenvalue of $\oper{A}$ for the observed system is read from \eqref{eq:unitary_evolution_interaction} to be $\lambda$. Therefore,
\begin{equation}
\label{eq:unitary_evolution_interaction_result}
\oper{U}_t\ket{\lambda,a}\ket{\tn{ready}} 
= \ket{\lambda,a}\ket{\tn{result}=\lambda}.
\end{equation}

\begin{lemma}
\label{lemma:measurement-transform}
Let $\oper{A}$ be an observable that has $-\lambda$ as an eigenvalue whenever $\lambda$ is an eigenvalue, and whose eigenspaces have equal dimension.
This includes the cases when the spectrum of $\oper{A}$ is a continuous interval $(-\lambda_{\tn{max}},+\lambda_{\tn{max}})$ or $[-\lambda_{\tn{max}},+\lambda_{\tn{max}}]$.
Then, there is a unitary transformation $\oper{S}$ of the total Hilbert space so that
\begin{equation}
\label{eq:qubit-change-Hamiltonian}
\oper{H}\oper{S}=\oper{S}\oper{H}
\end{equation}
and
\begin{equation}
\label{eq:qubit-change-condition}
\oper{S}\ket{\lambda,a}\ket{\tn{result}=\lambda}=\ket{-\lambda,a}\ket{\tn{result}=-\lambda}.
\end{equation}
\end{lemma}
\begin{proof}
In the standard measurement scheme, we choose the unitary transformation defined on the basis vectors of the Hilbert space by
\begin{equation}
\label{eq:qubit-change-def}
\oper{S}\ket{\lambda,a}\ket{\zeta}=\ket{-\lambda,a}\ket{-\zeta}.
\end{equation}

Then, from equation \eqref{eq:unitary_evolution_interaction} , condition \eqref{eq:qubit-change-condition} is satisfied.

For the condition \eqref{eq:qubit-change-Hamiltonian}, we notice that for all $t\in[0,T]$
\begin{equation}
\label{eq:interaction_commute_change}
\begin{aligned}
\oper{U}_t\oper{S}\ket{\lambda,a}\ket{\zeta} 
&\stackrel{\eqref{eq:qubit-change-def}}{=} \oper{U}_t\ket{-\lambda,a}\ket{-\zeta} \\
&\stackrel{\eqref{eq:unitary_evolution_interaction}}{=} \ket{-\lambda,a}\ket{-\zeta + g t \lambda}\\
&\stackrel{\eqref{eq:qubit-change-def}}{=}\oper{S}\ket{\lambda,a}\ket{\zeta - g t \lambda}\\
&\stackrel{\eqref{eq:unitary_evolution_interaction}}{=}\oper{S}\oper{U}_t\ket{\lambda,a}\ket{\zeta}.
\end{aligned}
\end{equation}

Therefore, for all $t\in[0,T]$,
\begin{equation}
\label{eq:unitary_commute_change}
\oper{U}_t\oper{S} = \oper{S}\oper{U}_t.
\end{equation}

By taking the limit $t\searrow 0$, it follows that $\oper{H}\oper{S}=\oper{S}\oper{H}$, so condition \eqref{eq:qubit-change-Hamiltonian} is satisfied too.
\end{proof}

\begin{proof}[Another proof (using eigenbases)]
We work in the basis of vectors of the form $\ket{\zeta}\ket{p}$, where $\ket{p}$ are eigenvectors of $\oper{p}_{\mathbf{Z}}$, $\oper{p}_{\mathbf{Z}}\ket{p}=p\ket{p}$, which from equation \eqref{eq:pointer_ccr} satisfy
\begin{equation}
\label{eq:fourier-p}
\braket{\zeta}{p}=\frac{1}{\sqrt{2\pi \hbar}} e^{ip\zeta/\hbar}.
\end{equation}

Then, the Hamiltonian \eqref{eq:measurement_hamiltonian} is diagonal in this basis,
\begin{equation}
\label{eq:measurement_hamiltonian_diag}
\oper{H}=\oper{H}_{\tn{int}} = -g \sum_{a,\lambda}\int_{-\infty}^\infty\lambda p\ket{\lambda,a}\ket{p}\bra{p}\bra{\lambda,a}\de p.
\end{equation}

We define the operator $\oper{S}$ by
\begin{equation}
\label{eq:qubit-change-def-diag}
\oper{S}\ket{\lambda,a}\ket{p}=\ket{-\lambda,a}\ket{-p}.
\end{equation}

This operator commutes with $\oper{H}$, because both $\ket{\lambda,a}\ket{p}$ and $\ket{-\lambda,a}\ket{-p}$ are eigenvectors of $\oper{H}$ with the same eigenvalue, $-g \lambda p=-g (-\lambda)(-p)$.
Using
\begin{equation}
\label{eq:fourier}
\ket{\zeta}=\frac{1}{\sqrt{2\pi \hbar}}\int_{-\infty}^\infty e^{-ip\zeta/\hbar}\ket{p}\de p
\end{equation}
one obtains
\begin{equation}
\label{eq:qubit-change-def-pZ}
\begin{aligned}
\oper{S}\ket{\lambda,a}\ket{\zeta}
&\stackrel{\eqref{eq:fourier}}{=}\oper{S}\ket{\lambda,a}\frac{1}{\sqrt{2\pi \hbar}}\int_{-\infty}^\infty e^{-ip\zeta/\hbar}\ket{p}\de p\\
&\stackrel{\eqref{eq:qubit-change-def-diag}}{=}\ket{-\lambda,a}\frac{1}{\sqrt{2\pi \hbar}}\int_{-\infty}^\infty e^{-ip\zeta/\hbar}\ket{-p}\de p\\
&\stackrel{\phantom{\eqref{eq:qubit-change-def-diag}}}{=}\ket{-\lambda,a}\frac{1}{\sqrt{2\pi \hbar}}\int_{-\infty}^\infty e^{+ip\zeta/\hbar}\ket{p}\de p\\
&\stackrel{\eqref{eq:fourier}}{=}\ket{-\lambda,a}\ket{-\zeta}.
\end{aligned}
\end{equation}

From this, one obtains \eqref{eq:unitary_evolution_interaction_result}.
\end{proof}

In the case when the spectrum of the observable $\oper{A}$ is $\R$, we also have the following result.
\begin{lemma}
\label{lemma:measurement-transform-swap}
If the spectrum of $\oper{A}$ is $\R$ and the eigenspaces have equal dimension, for any pair of non-null eigenvalues $\lambda_1\neq\lambda_2$ there is a unitary transformation $\oper{S}$ of the total Hilbert space so that $\oper{H}\oper{S}=\oper{S}\oper{H}$ and
\begin{equation}
\label{eq:qubit-change-condition-swap}
\oper{S}\ket{\lambda_1}\ket{\tn{result}=\lambda_1}=\abs{\frac{\lambda_2}{\lambda_1}}\ket{\lambda_2}\ket{\tn{result}=\lambda_2}.
\end{equation}
\end{lemma}
\begin{proof}
This is achieved by the unitary transformation
\begin{equation}
\label{eq:qubit-change-def-swap}
\oper{S}\ket{\lambda,a}\ket{\zeta}=\abs{\frac{\lambda_2}{\lambda_1}}\ket{\frac{\lambda_2}{\lambda_1}\lambda,a}\ket{\frac{\lambda_2}{\lambda_1}\zeta},
\end{equation}
where the coefficient $\abs{\frac{\lambda_2}{\lambda_1}}$ is due to the scaling property of the Dirac $\delta$ distribution.
Then,
\begin{equation}
\label{eq:interaction_commute_change-swap}
\begin{array}{rcl}
\oper{U}_t\oper{S}\ket{\lambda,a}\ket{\zeta}
&\stackrel{\eqref{eq:qubit-change-def-swap}}{=}& \abs{\frac{\lambda_2}{\lambda_1}}\oper{U}_t\ket{\frac{\lambda_2}{\lambda_1}\lambda,a}\ket{\frac{\lambda_2}{\lambda_1}\zeta} \\
&\stackrel{\eqref{eq:unitary_evolution_interaction}}{=}& \abs{\frac{\lambda_2}{\lambda_1}}\ket{\frac{\lambda_2}{\lambda_1}\lambda,a}\ket{\frac{\lambda_2}{\lambda_1}\zeta - g t \frac{\lambda_2}{\lambda_1}\lambda}\\
&\stackrel{\eqref{eq:qubit-change-def-swap}}{=}& \oper{S}\ket{\lambda,a}\ket{\zeta - g t \lambda}\\
&\stackrel{\eqref{eq:unitary_evolution_interaction}}{=}& \oper{S}\oper{U}_t\ket{\lambda,a}\ket{\zeta}.
\end{array}
\end{equation}

Therefore, $\oper{U}_t\oper{S} = \oper{S}\oper{U}_t$ for all $t\in[0,T]$ and the limit $t\searrow 0$ gives $\oper{H}\oper{S}=\oper{S}\oper{H}$.
\end{proof}

\vspace{-0.5cm}

%
%
%
%
%
%


\end{document}